\documentclass[11pt,reqno]{article}

\usepackage{PRIMEarxiv}

\usepackage{times,fancyhdr,paralist}
\usepackage{amsmath,amsthm,amssymb,amsfonts,mathrsfs}
\usepackage{amscd,bm}
\usepackage{amsxtra,dsfont,xspace,xfrac}
\usepackage{mathtools}
\usepackage{nth}
\usepackage[dvipsnames]{xcolor}
\usepackage{graphicx}
\usepackage{url}
\usepackage{booktabs}

\newcommand{\fat}[1]{\mathds{#1}}

\newcommand{\NN}{\fat{N}}
\newcommand{\ZZ}{\fat{Z}}

\newcommand{\RR}{\fat{R}}

\newcommand{\minus}{\smallsetminus}

\newcommand{\com}{^{{\ast}}}
\newcommand{\comp}{^{{\scriptscriptstyle\complement}}}

\newcommand{\norm}[1]{\big\Vert #1\big\Vert}

\newcommand{\diam}[1]{\mathtt{diam}\!\left(#1\right)}

\newcommand{\unitball}{\mathbb{B}}

\newcommand{\idm}[1]{{\mathbf{I}_{#1}}}

\newcommand{\conv}[1]{\mathtt{conv}\!\left(#1\right)}

\newcommand{\prob}[1]{{\fat{P}\!\left[#1\right]}}

\newcommand{\kldiv}[2]{D_{\!\scriptscriptstyle{K\!L}\!}\!\left(#1\big\Vert #2\right)}

\newcommand{\dg}[1]{\textcolor{red}{#1}}

\newcommand{\work}{\Omega}
\newcommand{\sh}{{}^{\scriptscriptstyle{\sharp}}}

\newcommand{\ee}[1]{\mathtt{e}^{#1}}

\newcommand{\bdd}{\bar{d}}
\newcommand{\goal}{x\com}
\newcommand{\upp}[1]{^{\scriptscriptstyle{(#1)}}}
\newcommand{\rsupp}[1]{{}^{\scriptscriptstyle{\sharp(#1)}}}

\newcommand{\wgt}[2]{{\omega_{#1}\upp{#2}}}
\newcommand{\rswgt}[2]{{\omega\rsupp{#2}_{#1}}}

\newcommand{\ts}[1]{{t_{#1}}}
\newcommand{\dts}[1]{\vartriangle\!\!t_{#1}}
\newcommand{\xhat}[2]{{\hat{x}_{#1}\upp{#2}}}

\newcommand{\intent}[2]{{\theta_{#1}\upp{#2}}}
\newcommand{\pt}[2]{{x_{#1}\upp{#2}}}

\newcommand{\rad}[2]{{r_{#1}\upp{#2}}}
\newcommand{\gt}[2]{{t_{#1}\upp{#2}}}

\newcommand{\gain}[2]{{\lambda_{#1}\upp{#2}}}

\newcommand{\qrep}[2]{{q_{#1}\upp{#2}}}

\newcommand{\wdim}{n}

\newcommand{\rmin}{r_{\scriptscriptstyle{m\!i\!n}}}
\newcommand{\rmax}{r_{\scriptscriptstyle{m\!a\!x}}}

\newcommand{\tmin}{T_{\scriptscriptstyle{m\!i\!n}}}
\newcommand{\tmax}{T_{\scriptscriptstyle{m\!a\!x}}}
\newcommand{\tacc}{T_{\scriptscriptstyle{acc}}}

\newcommand{\effsym}{\scriptscriptstyle{e\!f\!f}}
\newcommand{\neff}{{N_{\effsym}}}
\newcommand{\wgteff}{{\omega_{\effsym}}}
\newcommand{\rswgteff}{{\omega\sh_{\effsym}}}
\newcommand{\htheta}{{\hat{\theta}}}

\newcommand{\avgsym}{{\scriptscriptstyle{\mathtt{ }}}}
\newcommand{\redsym}{{\scriptscriptstyle{\mathtt{red}}}}
\newcommand{\particledistro}{\boldsymbol{\pi}}

\newcommand{\hqred}{\hat{q}^\redsym}
\newcommand{\hqavg}{\hat{q}^\avgsym}
\newcommand{\rshqred}{\hat{q}^{\scriptscriptstyle{\sharp}\redsym}}
\newcommand{\rshqavg}{\hat{q}^{\scriptscriptstyle{\sharp}\avgsym}}

\newcommand{\errcov}{{\boldsymbol{\Delta}}}
\newcommand{\discov}{{\boldsymbol{\Sigma}}}
\newcommand{\pcov}[1]{{\mathbf{P}}_{#1}}
\newcommand{\kgain}[1]{{\mathbf{K}}_{#1}}

\newcommand{\major}[1]{\mathcal{M}_{#1}}
\newcommand{\rsmajor}[1]{\mathcal{M}\sh_{#1}}
\newcommand{\allparticles}{\mathcal{P}}

\newcommand{\hu}{\underline{H}}

\newcommand{\rshavg}[1]{H_{#1}^{\scriptscriptstyle{\sharp}\avgsym}}

\newcommand{\rshavglow}[1]{\underline{H}_{{#1}}^{\scriptscriptstyle{\sharp}\avgsym}}

\newcommand{\Fsharp}{F^{\scriptscriptstyle{\sharp}}}
\newcommand{\zsharp}{z^{\scriptscriptstyle{\sharp}}}
\newcommand{\CIbeta}{C_I^{\beta}}

\theoremstyle{plain}
\newtheorem{theorem}{Theorem}

\theoremstyle{definition}

\newtheorem{assumption}{Assumption}
\newtheorem{remark}{Remark}
\newtheorem{lemma}{Lemma}

\title{Quantifying Trade-offs Between Stability and Goal-Obfuscation%
\thanks{This research was supported in part by Air Force Office of Scientific Research award number FA9550-22-1-0429. Any opinions, findings and conclusions or recommendations expressed in this material are those of the author(s) and do not necessarily reflect the views of the sponsoring agency.}}

\author{%
Yixuan Wang$^{1}$\thanks{Corresponding author: \texttt{wang.yixuan@ufl.edu}.}
\quad
Dan P.~Guralnik$^{2}$
\quad
Warren E.~Dixon$^{1}$
\\[0.6em]
{\normalfont $^{1}$Department of Mechanical and Aerospace Engineering}
\\
{\normalfont University of Florida}
\\
{\normalfont\texttt{\{wang.yixuan,wdixon\}@ufl.edu}}
\\
{\normalfont $^{2}$Mathematics Department, Ohio University}
\\
{\normalfont\texttt{danguralnik@ohio.edu}}
}

\begin{document}

\maketitle
\thispagestyle{empty}
\pagestyle{empty}

\begin{abstract}
Safety-critical autonomy in adversarial settings demands more than Lyapunov stability of tracking error signals.
An agent executing a goal-directed trajectory is intrinsically legible to a passive observer running online Bayesian inference, because the contractive dynamics of any Lyapunov basin of attraction concentrates posterior belief over the latent intent parameters.
This article initiates the study of intent privacy over a continuous state space as a joint control problem on the physical state combined with the latent belief state of a putative observer.

With the main challenges concentrated around the analysis of the belief-state dynamics, the agent dynamics is assumed to be simple, modeled by the differential inclusion $\dot{x}\in u+\bar{d}\mathbb{B}$.
That is, the agent is fully actuated with bounded unknown disturbance to the control input. 
The observer's intent inference process is modeled as a discrete-time stochastic dynamical system evolving over the belief state space of a Rao--Blackwellized particle filter (RBPF) reasoning over large random samples of possible agent goals.
The agent's control input is modeled as a piecewise constant signal, with jumps matching the RBPF update times.
Building on a prior intent-inference framework and its online-computable, KL-based information leakage measurement, a privacy constraint is imposed, which amounts to maintaining information leakage above a prescribed threshold with high probability, using probabilistic discrete-time control barrier functions (PCBFs).
A key technical contribution is the derivation of separate PCBF results for the Bayesian update step and the resampling step of the RBPF, enabling a PCBF result for the full update as well as integration of the privacy constraint with the agent's task-side tracking requirement.
%
%
%
Finally, a joint feasibility analysis is carried out by examining the interplay between the privacy constraint and the tracking envelope.
%
\end{abstract}

\section{Introduction}
Control Lyapunov functions and control barrier functions provide a standard mechanism for synthesizing real-time safety critical controllers through online convex programs~\cite{ames2017cbf, ames2019cbf_survey}.
In this paradigm, a controller is minimally perturbed at each timestep so that the resulting input simultaneously drives a Lyapunov function toward zero and keeps a barrier function nonnegative, yielding joint guarantees of stability and constraint satisfaction through a single quadratic program~\cite{ames2017cbf}.
The framework has been deployed across a wide range of domains including legged locomotion, adaptive cruise control, and multi-robot coordination~\cite{ames2019cbf_survey}.

An implicit assumption underlying Lyapunov-based controller synthesis is that closed loop stability is decidedly a desirable property.
%
%
Classically, the structural aim of Lyapunov-based design is to produce closed loop dynamics with normally hyperbolic invariant manifold (NHIM) structure~\cite{wiggins1994nhim}, in which the basin of attraction decomposes as a product of the goal-parameterized attractor with a normal fiber on which the dynamics acts contractively.
Consequently, if the attractor structure is known to (or could be inferred by) an external observer up to a finite-dimensional parameter vector, the contractive geometry of the closed loop trajectories concentrates the observer's posterior over those parameters with increasing speed as the agent converges to the goal.
For a goal-directed agent, the relevant latent parameters are the goal location, goal radius, and desired arrival time, and this legibility problem is exacerbated by the fact that task optimal trajectories are precisely those most concentrated near the attracting manifold.
The work of~\cite{dragan2013legibility} showed that trajectories optimizing a task cost are inherently legible: a Bayesian observer can rapidly concentrate posterior belief over the agent's goal from a short observation window, and intentional departure from task optimality is necessary to reduce this legibility.
Complementary work treats the observer's inference problem as inverse planning~\cite{ramirez2009plan_recognition, baker2009action} or maximum entropy inverse reinforcement learning~\cite{ziebart2008maxent}, confirming that goal-stabilizing trajectories are informationally richest in precisely the regimes where Lyapunov methods are most effective.

Prior work~\cite{wang2025goal} formalized and quantified this legibility for continuous-state systems of the form studied here.
There, the adversary models the agent's intent and maintains $N\gg 0$ weighted particles to recursively estimate the agent's intent from noisy position measurements via a Rao-Blackwellized particle filter (RBPF), demonstrating that under nominal conditions the adversary rapidly recovers the agent's true intent, underscoring the practical severity of the legibility problem for any agent operating under Lyapunov-based control.

The present paper addresses the control synthesis question left open by~\cite{wang2025goal}:
given a capable adversary running the RBPF inference framework, how should the agent select its control input to delay inference concentration while simultaneously satisfying task stability and tracking requirements?

\medskip
\noindent\textbf{Contributions.}
The contributions of this paper are as follows.
First, the intent privacy task is formulated as a piecewise constant (intermittent actuation) control problem over the latent belief dynamics on the probability simplex induced by RBPF weight updates (extended by the RBPF's state estimation variables), rather than as an auxiliary objective on the physical state.

Second, a finite horizon privacy-in-probability guarantee for the RBPF update cycle is derived.
The technical challenge of analyzing the resampling step is overcome by observing that the RBPF update may be regarded as the composition of two probabilistically independent random transformations of the belief state space: the Bayesian update, and the resampling step.
Probabilistic bounds on the increment in information leakage are derived for each of these transformations, and combined into a PCBF result for the overall RBPF update, following the PCBF framework developed in~\cite{mestres2025probabilistic}.
It is noteworthy that the bound for the resampling transformation is completely independent of the dynamics and control input, motivating future work on extensions to more complex dynamics of the agent.

Third, a joint feasibility analysis is carried out, examining the interplay between the privacy constraint and the tracking constraint via affine interpolation between two types of control input, one selected for obfuscation and the other for optimized tracking.
The tensions between the constraints are revealed as inequality constraints on the interpolation parameter, which may be seen to be jointly feasible when the prescribed tracking error envelope is not overly restrictive.
Further study of convergence properties of the RBPF is necessary to produce results establishing the infeasibility of obfuscation if practical stability of tracking errors is enforced.

\section{Background}
This article leans mainly on development from two sources.
Section~\ref{intent inference} reviews~\cite{wang2025goal}, where RBPFs for intent inference were introduced.
Section~\ref{sec:pre pcbf} reviews necessary results from~\cite{mestres2025probabilistic} on probabilistic barrier functions.

\subsection{Intent Inference}\label{intent inference}
\noindent{\bf Intent Model.} The agent's intent is modeled as a triple $\theta\com\triangleq(\goal,r\com,t\com)$, constrained to a compact domain $\Theta$, where $\goal\in\RR^n$ is the goal center,
$r\com > 0$ is the goal radius, and $t\com>0$ is the desired arrival time; the corresponding goal set is the ball $\goal+r\com\unitball$, which the agent aims to reach by time $t\com$.
The adversary assumes that, if the (unknown) intent were $\theta$, then the agent's closed-loop motion is goal-stabilizing up to an unknown disturbance bounded in norm by a known $\Bar{d}$:
\begin{equation}
  \dot{\hat{x}} \in f_\theta(\hat{x}) + \bar{d}\unitball,\quad
  f_\theta(\hat{x})\triangleq -\lambda(\theta)(\hat{x} - x(\theta)),\nonumber
\end{equation}
where $\lambda(\theta) \triangleq \max\left\{
  \frac{\bar{d}}{r(\theta)}, 
  \frac{1}{t(\theta)}\log\frac{R}{r(\theta)}
\right\}$,
ensuring that the agent will enter the goal ball $\goal+r\com\unitball$ by time $t\com$.\\

\noindent{\bf RBPF State Representation.}
The adversary maintains $N$ weighted particles $\{(\intent{k}{i}, \wgt{k}{i})\}_{i=1}^N$ to estimate $\theta^*$.
We distinguish three stages within each timestep $k$:
\begin{enumerate}
  \item \emph{Previous step weights} $\wgt{k-1}{i}$: weights at the end of timestep $k-1$ (post resampling if resampling occurred).
  \item \emph{Pre-resampling weights} $\rswgt{k}{i}$: weights after Bayesian update at timestep $k$, before resampling.
  \item \emph{Post resampling weights} $\wgt{k}{i}$: final weights at timestep $k$ after resampling (if triggered).
\end{enumerate}

\textbf{Initialization.} The $N$ particles are assigned uniform weights $\wgt{0}{i} \triangleq 1/N$, with positions $\pt{0}{i}$ sampled uniformly from the workspace $\Omega$, radii $\rad{0}{i}$ sampled uniformly from $[\rmin, \rmax]$, and times $\gt{0}{i}$ sampled uniformly from $[\tmin, \tmax]$.
The corresponding product distribution is denoted by $\particledistro$.
The initial state estimate corresponding to each particle $i$ is set to the initial measurement: $\xhat{0}{i} := y_0$.

\textbf{Propagation Step.} Let $t_k$ denote the discrete observation update times and let $\dts{k}\triangleq t_k-t_{k-1}$.
Throughout this paper, we assume that $\dts{k} = \Delta t$ for all $k$, where $\Delta t$ is a constant.
Discrete Euler integration of the dynamics provides the prior $\xhat{k}{i-}$:
\begin{align*}
  \xhat{k}{i-} = \xhat{k-1}{i} + \dts{k} f_{\intent{k-1}{i}}(\xhat{k-1}{i}) + \xi_k, \quad \xi_k \sim \mathcal{N}(0, \dts{k}^2 \discov),
\end{align*}
where $\discov \triangleq (\sigma \bdd)^2 I_\wdim$ and $\sigma > 0$ is a user selected parameter.
The prior error covariance and Kalman gain are
\begin{align*}
  \pcov{k}^- = (\gain{k-1}{i})^2 \pcov{k-1} + \discov, \quad \kgain{k} = \pcov{k}^- (\pcov{k}^- + \errcov)^{-1},
\end{align*}
yielding the state posterior and posterior error covariance:
\begin{align*}
  \hat{x}_k^{\rsupp{i}} = \xhat{k}{i-} + \kgain{k}(y_k - \xhat{k}{i-}), \quad \pcov{k} = (I - \kgain{k})\pcov{k}^-.
\end{align*}

\textbf{Bayesian Update Step.} Given observation $y_k$ at timestep $k$, the particle weights are updated via
\begin{align}\label{eqn:weight update}
  \rswgt{k}{i} = \frac{\wgt{k-1}{i} p(y_k |  \hat{x}_k\rsupp{i})}{\sum_{j=1}^N \wgt{k-1}{j} p(y_k |  \hat{x}_k\rsupp{j})},
\end{align}
where the likelihood is Gaussian: $p(y_k | \hat{x}_k^{\rsupp{i}}) \sim \mathcal{N}(\hat{x}_k^{\rsupp{i}}, \errcov)$ with symmetric positive definite covariance $\errcov$.

\textbf{Resampling Step.} When the effective sample size
\begin{align}\label{eqn:effective sample size}
  \neff \triangleq \left\lfloor \left(\textstyle\sum_{i=1}^N (\rswgt{k}{i})^2\right)^{-1} \right\rfloor
\end{align}
drops below a threshold $N_0$, resampling is triggered.
Let $\rsmajor{k} \subseteq \allparticles \triangleq \{1,\ldots,N\}$ denote the indices of the $\neff$ highest weight particles.
Define the effective mass before resampling as
\begin{align*}
  \rswgteff \triangleq \textstyle\sum_{a \in \rsmajor{k}} \rswgt{k}{a}.
\end{align*}
In~\cite[Proposition 1]{wang2025effective}, it is shown that
\begin{align*}
    1-\rswgteff\leq\frac{N-\neff}{N}\left(
        1-\sqrt{\frac{N-\neff-1}{(\neff+1)(N-1)}}
    \right).
\end{align*}
Each retained particle $a \in \rsmajor{k}$ is replicated $N_a \triangleq \lfloor \rswgt{k}{a} N / \rswgteff \rfloor$ times.
Let $\major{k}$ denote the set of indices of replicated effective particles. 
$\major{k}$ has cardinality $\sum_{a \in \rsmajor{k}} N_a$.
The remaining particles are reinitialized by sampling from a distribution $\particledistro$ over $\Theta$.
After resampling, all weights are distributed uniformly, $\wgt{k}{i} = \tfrac{1}{N}$ for all $i\in\allparticles$.
If no resampling is triggered ($\neff \geq N_0$), then set $\major{k}:=\rsmajor{k}$ and $\wgt{k}{i} := \rswgt{k}{i}$ for all $i \in \allparticles$.
%

\begin{remark}
\label{remark}
  The intent parameter $\intent{k+1}{i} = \intent{k}{i}$ remains constant unless the $i^\mathrm{th}$ particle is discarded during resampling.
  In other words, the intent parameter attached to any particle is held constant throughout the particle's lifetime.
  Similarly, $\xhat{k}{i}=\hat{x}_k^{\rsupp{i}}$ for all $i$ and $k$, if no resampling occurs.
  When resampling is triggered at timestep $k$, each retained particle inherits the state estimate and error covariance from its parent, and each reinitialized particle $b$ draws a fresh intent parameter value $\intent{k}{b}\sim\particledistro$ and a fresh initial position $\xhat{k}{b}$ drawn uniformly from the workspace $\work$.
\end{remark}

\noindent{\bf Two estimators. } At any time $t=\ts{k}$, the $N$ weighted particles $(\intent{k}{i},\wgt{k}{i})_{i=1}^N$ give rise to estimators of $\theta\com$, as follows.
Any such estimator, ultimately, is a random variable, which may be regarded as a measurement of $\theta\com$.
Thus, each $\theta\in\Theta$ ought to be represented by a probability density $q_\theta$ over $\Theta$ with mean $\theta$.
Following our earlier convention, denote $\qrep{k}{i}\triangleq q_{\intent{k}{i}}$, for simplicity.
For  simplicity, the $q_\theta$ are selected as product distributions over $\widetilde\Theta\triangleq\RR^\wdim\times(0,\infty)^2$ of the form
\begin{equation}\label{eqn:probabilistic representation of intent}
  \begin{aligned}
    &q_\theta \sim \mathcal{N}(x(\theta),\sigma_x^2\idm{\wdim}) \otimes \mathcal{R}\otimes \mathcal{T},\\
    &\log(\mathcal{R}) = \mathcal{N}(r(\theta), \sigma_r^2),\\
    &\log(\mathcal{T}) = \mathcal{N}(t(\theta), \sigma_t^2),
  \end{aligned}
\end{equation}
where $\sigma_x,\sigma_r,\sigma_t$ are system parameters representing the capacity of the adversary for precise localization and timing and guaranteeing that an overwhelming probability is assigned to the event $\Theta\subset\widetilde\Theta$.
The complete and reduced estimators, pre-resampling, are defined as
\begin{equation}\label{eqn:weighted avg estimator pre resampling}
  \rshqavg_k\triangleq\textstyle\sum_{i\in\allparticles} \rswgt{k}{i}\qrep{k-1}{i},\quad
  \rshqred_k\triangleq\textstyle\sum_{i\in \rsmajor{k}}\frac{\rswgt{k}{i}}{\wgteff}\qrep{k-1}{i}.\nonumber
\end{equation}
If no resampling takes place, then $\qrep{k}{i}=\qrep{k-1}{i}$ for all $i\in\allparticles$ and we obtain
\begin{equation}\label{eqn:weighted avg estimator}
  \hqavg_k\triangleq\textstyle\sum_{i\in\allparticles} \wgt{k}{i}\qrep{k}{i},\quad
  \hqred_k\triangleq\textstyle\sum_{i\in \major{k}}\frac{\wgt{k}{i}}{\wgteff}\qrep{k-1}{i},
\end{equation}
where $\wgteff\triangleq\sum_{a\in\major{k}}\wgt{k}{i}$.
If resampling takes place then, taking note of the re-enumeration of all particles, Equation~\eqref{eqn:weighted avg estimator} is replaced by
\begin{equation}\label{eqn:weighted avg estimator post resampling}
  \hqavg_k\triangleq\textstyle\tfrac{1}{N}\sum_{i\in\allparticles}\qrep{k}{i},\quad
  \hqred_k\triangleq\textstyle\tfrac{1}{N}\sum_{i\in \major{k}}\qrep{k}{i}.
\end{equation}
In this paper, it is sufficient to employ the complete estimator, since the variation in the information lower bound across different estimators is only an additive constant.

\noindent{\bf KL-based information leakage measurement.}
Let $H_{k}\sh$ and $H_{k}$ denote the information leakage from the true intent distribution to the estimated one at time step $k$ using the complete estimators pre-resampling and post-resampling, respectively,
%
\begin{equation}\label{eqn:information leakage definition}
\begin{aligned}
  H_{k}\triangleq\kldiv{q_{\theta\com}}{\hqavg_k},\quad H_{k}^{\scriptscriptstyle{\sharp}}\triangleq\kldiv{q_{\theta\com}}{\rshqavg_k}.\nonumber
\end{aligned}
\end{equation}
\begin{theorem}\label{thm:bound} Let $S_{\nu, k}^{\scriptscriptstyle{\sharp}}$ denotes the weighted exponential sum before resampling
\begin{align*}
    S_{\nu, k}^{\scriptscriptstyle{\sharp}} &\triangleq
        \textstyle\sum_{j \in \allparticles}\rswgt{k}{j}\gamma_\nu(\intent{k}{j}),
\end{align*}
where
\begin{align*}
    \gamma_\nu(\theta) \triangleq \ee{-\tfrac{\norm{\nu(\theta\com)- \nu(\theta)}^2}{4\sigma_{\nu}^2}}
\end{align*}
is the kernel.
The following lower bounds hold:
  \begin{equation}\label{Info lower bound}
    \begin{aligned}
      \rshavg{k} &\geq \rshavglow{k} \triangleq-\log\left(
        \sigma_x^{1-\wdim}(\frac{\ee{}}{2})^{\frac{\wdim+2}{2}}
      \right)-\sum_{\nu\in\{x,r,t\}}\log S_{\nu,k}^{\scriptscriptstyle{\sharp}}.\nonumber
    \end{aligned}
  \end{equation}
\end{theorem}
The departure point of this paper is the crucial observation that the lower bounds produced by Theorem~\ref{thm:bound} are time-independent functions of the information state of the putative adversary's RBPF.
As a result, these functions could be put to use as control-barrier functions.

\subsection{Probabilistic safety and PCBF}
\label{sec:pre pcbf}
Consider a discrete-time controlled stochastic system
\begin{equation}
  x_{t+1} = F(x_t,u_t,d_t), \quad t\in\ZZ_{\ge 0},\nonumber
\end{equation}
where $x_t\in\mathcal{X}\subseteq\RR^n$ is the state, $u_t\in\mathcal{U}\subseteq\RR^m$ is the control input constrained to an admissible set $\mathcal{U}$, and $d_t\in\mathcal{D}$ is an exogenous random disturbance.
We assume $\{d_t\}_{t\geq 0}$ is an i.i.d. sequence with common distribution $\mathcal{P}_d$ (equivalently, $d_t\sim \mathcal{P}_d$).
Unless stated otherwise, probabilities below are taken with respect to $\mathcal{P}_d$ (and any additional exogenous randomness entering $F$), conditioned on the current $(x_t,u_t)$.

{\bf Safe set and finite-horizon probabilistic safety.} Let $h:\mathcal{X}\to\RR$ be a measurable function defining the safe set
\begin{equation}
  C := \{x\in\mathcal{X}:\ h(x)\ge 0\}.\nonumber
\end{equation}
Given a horizon $H\in\NN$ and a failure tolerance $\varepsilon\in(0,1)$, finite-horizon probabilistic safety requires
\begin{equation}
  \prob{x_t\in C \middle| x_0}\geq 1-\varepsilon \quad \text{for all } t=0,1,\dots,H.
  \label{eqn:finite_horizon_prob_safety}
\end{equation}

{\bf One-step PCBF condition.} Fix $\alpha\in[0,1]$ and $\delta\in(0,1)$.
Following \cite{mestres2025probabilistic}, the function $h$ is called a $\delta$-PCBF if for every $x\in C$ there exists an admissible input $u=u(x)\in\mathcal{U}$ such that
\begin{equation}
  \prob{h \left(F(x,u,d)\right) \geq \alpha h(x) \big| x,u} \geq 1-\delta,
  \label{eqn:pcbf_one_step}
\end{equation}
where the probability is taken over $d\sim\mathcal{P}_d$.
In online control, \eqref{eqn:pcbf_one_step} is enforced at time $t$ with the chosen input $u_t$:
\begin{equation}
  \prob{h(x_{t+1}) \geq \alpha h(x_t) \middle| x_t,u_t} \geq 1-\delta.
  \label{eqn:pcbf_one_step_online}
\end{equation}

{\bf From one-step to horizon guarantee.} Let $\pi:\mathcal{X}\to\mathcal{U}$ be a feedback policy and set $u_t=\pi(x_t)$.
If \eqref{eqn:pcbf_one_step_online} holds for all $t=0,1,\dots,H-1$ whenever $x_t\in C$, then the probability of remaining in $C$ for $H$ steps is lower-bounded by
\begin{equation}
  \prob{x_t\in C\middle|x_0\in C}\geq(1-\delta)^H \text{ for all } t=0,1,\dots,H.\nonumber
\end{equation}
Consequently, to satisfy \eqref{eqn:finite_horizon_prob_safety} it suffices to choose $\delta$ such that $(1-\delta)^H\geq 1-\varepsilon$, i.e.,
\begin{equation}
  \delta \leq 1-(1-\varepsilon)^{1/H}.\nonumber
\end{equation}

\section{Problem Formulation}
We consider two principal entities, an adversary and an agent, whose positions evolve in $\RR^\wdim$, $\wdim\in\{2,3\}$, corresponding to planar ($n=2$) or spatial ($n=3$) motion.
The agent has the goal of reaching the ball $\goal+r\com\unitball$ within a time $t\com$.
The goal must be reached without disclosing the intent parameter $\theta\com$ to the adversary, while maintaining a time-dependent bound $\varrho(t)$ on the error $\norm{x(t)-x_d(t)}$, where $x_d$ is determined by $\theta\com$ and the initial position $x_0\in\work$ of the agent.
The adversary’s objective is to infer the agent’s intent parameter $\theta\com$ by observing the agent’s trajectory.
The adversary is assumed to be a passive observer, focusing the analysis on the agent's control decisions aimed at obscuring $\theta\com$ from the adversary.
Note that the adversary is modeled internally by the agent solely as a proxy for intent-inference risk.
Accordingly, we do not explicitly model the adversary’s observation sampling rate, nor do we address the question of how the agent would access or verify the adversary’s inferred belief in a real deployment.

For simplicity, the agent is assumed to be fully actuated, but subject to bounded disturbances,
\begin{align}\label{eqn:dynamics}
    \dot{x}\in u+\bar{d}\unitball,
\end{align}
where the state $x\in\RR^\wdim$ is assumed to be measurable by the agent, $u\in\mathcal{U}\triangleq\RR^n$ is the agent's control input, and $\bar{d}>0$ is a known constant.
The task is encoded by the parameter $\theta\com$, whose values are confined to the domain
\begin{equation}
    \Theta\triangleq R\unitball\times[\rmin,\rmax]\times[\tmin,\tmax],\nonumber
\end{equation}
where $R>0$, $0<\rmin\ll\rmax\ll R$, and $0<\tmin\ll\tmax$ are parameters of the problem, fixed in advance.
The notation $\theta=(x(\theta),r(\theta),t(\theta))$ for $\theta\in\Theta$ will be used when there is no risk of ambiguity.
For each $\theta\in\Theta$ and $q\in\RR^\wdim$, a smooth reference path
\begin{align*}
    \label{eqn:reference paths}
    &x_{q,\theta}:[0,\infty)\to\RR^\wdim, 
    x_{q,\theta}(0)=q,  x_{q,\theta}(t(\theta))=x(\theta),
\end{align*}
is available to the agent, whose task, given $x(0)=q$, is to track this path by selecting control inputs $u(x,t)$ such that the tracking error $e(t)\triangleq x(t)-x_{q,\theta}(t)$ satisfies
\begin{equation}\label{eqn:envelope constraint}
    \norm{e(t)}\leq\varrho(t), 
    \varrho(t(\theta))=r(\theta)
\end{equation} 
for some increasing function $\varrho:[0,\infty)\to[0,\infty)$.
In an unobstructed (homogeneous and isotropic) environment, it is natural to select
\begin{equation}\label{eqn:reference path:interval}
    x_{q,\theta}(t):=q+\tfrac{t}{t(\theta)}(x(\theta)-q)
\end{equation}

The agent is assumed to be aware of the presence of an adversary attempting to infer the intent parameter $\theta\com$ from its behavior.
Therefore, in addition to~\eqref{eqn:envelope constraint}, the agent is tasked with selecting a controller whose resulting trajectory maintains the adversary's estimate $\htheta$ of the task parameter $\theta\com$ at an acceptable level of ambiguity within a time $\tacc<t(\theta)$.

\section{Probabilistic Information Safety Bounds}\label{sec:pcbf_privacy}
This section applies the PCBF framework to the RBPF belief state update to obtain finite-horizon privacy-in-probability guarantees.
\subsection{Information-state dynamics}\label{subsec:info_state_dynamics}
The observation $y_k$ received by the adversary at each timestep $k$ is generated by the agent's physical state, which is shaped by the control input $u$ through the dynamics~\eqref{eqn:dynamics}.
Consequently, the agent's choice of $u$ influences the RBPF belief update through the induced observation distribution, motivating a formulation in which $u$ appears explicitly as an argument of the belief transition.
Define the information state (of the RBPF) as\footnote{The information state of the RBPF at time $k$, in fact, contains also the covariance matrices $\mathbf{P_k}$, which may be treated as latent variables in the current application since the quantities of interest do not depend on them.}
\begin{equation}
    z \triangleq (\boldsymbol{\theta},\hat{\mathbf{x}},\omega)\in\mathcal{Z},\nonumber
\end{equation}
where $\boldsymbol{\theta}=(\theta\upp{i})_{i\in\allparticles}$ is the particle data, $\hat{\mathbf{x}}$ is the vector of state estimates associated with each particle, and $\omega$ is the RBPF weight vector.
%
%
The information dynamics is represented in controlled stochastic form as
\begin{equation}\label{eqn:info_dynamics_controlled}
    z' = F_I(z,u,\eta),
\end{equation}
where $u \in\RR^n$ is the agent's control input, $\eta$ aggregates all exogenous randomness driving the belief update including the observation noise $\xi$ and the resampling uncertainty $\zeta$, and 
$z'=(\boldsymbol{\theta}',\hat{\mathbf{x}}',\omega')$ denotes the successor information state.
As long as the agent's control input is decided independently of history, the following assumption is merited.
\begin{assumption}[Controlled Markov Property]\label{assump:markov_info}
The successor information state $z'$ is a random variable induced by~\eqref{eqn:info_dynamics_controlled}, and the sequence $\{\eta_k=\eta_k(t_k)\}_{k\ge 0}$ is i.i.d.
\end{assumption}
Consequently, when evaluated along a trajectory $z_k \triangleq z(t_k)$ at the $k$-th observation time, the sequence $\{z_k\}_{k\ge 0}$ is a controlled Markov process with transition kernel $\prob{z'\in \cdot\mid z,u}$.
\begin{remark}\label{rem:time_independent}
The subscript $k$ appears hereafter only when referring to the value of a quantity along a specific trajectory; all maps, barrier functions, and constraint sets considered in Sections~\ref{sec:pcbf_privacy} and~\ref{discussion} are time independent functions of $z$.
\end{remark}

To perform the analysis, the transition is decomposed as
\begin{equation}
  F_I(z,u,\eta) = \mathcal{R}(F^{\scriptscriptstyle{\sharp}}(z,u,\xi),\zeta),\nonumber
\end{equation}
where $F\sh$ is the Bayesian update producing the pre-resampling information state, $\xi$ is the observation noise, and $\zeta$ is the resampling noise.
The intermediate state $F\sh(z,u,\xi)$ has the form
\begin{equation}
  z^{\scriptscriptstyle{\sharp}}\triangleq (\boldsymbol{\theta},\hat{\mathbf{x}}^{\scriptscriptstyle{\sharp}},\omega^{\scriptscriptstyle{\sharp}}),\nonumber
\end{equation}
noting that the Bayesian update of the RBPF has no effect on the intent parameters.
The information leakage at state $z$ is equal to 
\begin{align*}
    H(z)\triangleq\kldiv{q_{\theta\com}}{\textstyle\sum_{i\in\allparticles}\wgt{}{i}q_{\theta\upp{i}}}.
\end{align*}
It is helpful to define the quantities
\begin{align*}
  S_\nu(z) &\triangleq \textstyle\sum_{j\in\allparticles} \wgt{}{j}\gamma_\nu(\theta\upp{j}),\quad
\end{align*}
where $\theta\upp{j}$ is the $j$-th particle held in the information state $z$.
Note that $S_\nu(z)$ is a function of the information state alone.
The post-resampling and pre-resampling information leakage lower bounds from Theorem~\ref{thm:bound} then take the time-independent forms
\begin{align*}
  H(z)\geq\hu(z) &\triangleq C - \textstyle\sum_{\nu\in\{x,r,t\}}\log S_\nu(z),
\end{align*}
where $C \triangleq -\log \left(\sigma_x^{1-\wdim}\left(\tfrac{\ee{}}{2}\right)^{\frac{\wdim+2}{2}}\right)$ is a constant depending only on system parameters.
Then, $\hu(z_k)\leq H_k$, $\hu(z^{\scriptscriptstyle{\sharp}}_k)\leq H^{\scriptscriptstyle{\sharp}}_k$ hold pointwise along any trajectory.
Similarly,
\begin{align*}
    H(z)&= \textstyle\int q_{\theta\com}\log \frac{q_{\theta\com}}{\sum_j \wgt{}{j} q\upp{j}}
    \leq \int q_{\theta\com} (-\sum_j \wgt{}{j}\log \frac{q\upp{j}}{q_{\theta\com}})\\
    &=\textstyle\sum_j \wgt{}{j}\int q_{\theta\com} \log \frac{q_{\theta\com}}{q\upp{j}}=
    \sum_j \wgt{}{j}\kldiv{q_{\theta\com}}{q\upp{j}}.
\end{align*}
Rewriting the right-hand side explicitly as
\begin{align*}
    H(z)\leq 
    \underbrace{\textstyle\sum_{j\in\allparticles} \omega^{(j)}
    \underbrace{\textstyle\sum_{\nu\in\{x,r,t\}}
    \tfrac{\norm{\nu(\theta^*) - \nu(\theta^{(j)})}^2}{2\sigma_\nu^2}}_{\kldiv{q_{\theta\com}}{q\upp{j}}}}_{\triangleq\bar{H}(z)}\leq M_H,
\end{align*}
it follows that $\bar{H}(z)$ is a continuous function of only the weights $\omega$ and the particles $\theta\upp{j}\in\Theta$, all confined to compact sets.
It follows that $\bar{H}(z)\leq M_H$ for $M_H>0$ big enough.

\subsection{Bayesian Update Barrier}
\label{sec:bayesian}
Fixing $\gamma > 0$, the information barrier with threshold $\gamma$ and the corresponding safe set are defined as
\begin{equation}
\begin{aligned}
    &b(z) \triangleq \hu(z) - \gamma,\\
    &C_I \triangleq \{z \in \mathcal{Z} : b(z) \geq 0\},\nonumber
\end{aligned}
\end{equation}
where the threshold $\gamma$ is a user-provided parameter which must satisfy $\gamma\ll M_H$ to guarantee $b(z)>0$ at initialization.

We establish a one step PCBF condition for the Bayesian
update map $\Fsharp$.
The barrier change is
\begin{equation}\label{eqn:barrier_change}
  b(\zsharp) - b(z)
  = \textstyle\sum_{\nu \in \{x,r,t\}}
    \log \tfrac{S_\nu(z)}{S_\nu(\zsharp)}.
\end{equation}
Define the likelihood ratio
\begin{align*}\label{eqn:observation likelihood ratio}
    r\upp{j}(y) \triangleq
p(y|\xhat{}{j})\big{/}\textstyle\sum_i \wgt{}{i} p(y|\xhat{}{i}).
\end{align*}
Note that $\rswgt{}{j} = \wgt{}{j} r\upp{j}$ and
$\sum_j \wgt{}{j} r\upp{j} = 1$, by~\eqref{eqn:weight update}.
Then, the ratio
$S_\nu(\zsharp)/S_\nu(z)$ lies in the interval $[\min_j r\upp{j},\max_j r\upp{j}]$ as a convex combination of the
$r\upp{j}$.
It follows from~\eqref{eqn:barrier_change} by summing over $\nu\in\{x,r,t\}$ that
\begin{equation}\label{eq:rsp_bound}
  |b(\zsharp) - b(z)| \leq 3B(z,y),
\end{equation}
where the bound $B(z,y)$ is defined as
\begin{align*}\label{eqn:log bound on barrier change}
    B(z,y)\triangleq
    \textstyle\max \bigl(\vert\log\max_j r\upp{j}(y)\vert,\;
    \vert\log\min_j r\upp{j}(y)\vert\bigr).
\end{align*}

To bound $B(z,y)$ as a function of $y$, we compute
$\nabla_y \log r\upp{j}(y)$.
Since $p(y|\xhat{}{j}) \sim \mathcal{N}(\xhat{}{j}, \errcov)$
by~\eqref{eqn:weight update}, a direct computation gives
\begin{equation}
  \nabla_y \log r\upp{j}(y)
  = \errcov^{-1} 
    \bigl(\xhat{}{j} - \mu^{\scriptscriptstyle{\sharp}}(y)\bigr),\nonumber
\end{equation}
where $\mu^{\scriptscriptstyle{\sharp}}(y) \triangleq
\sum_i \rswgt{}{i}(y) \xhat{}{i}$
is the posterior weighted mean.
Since
$\mu^{\scriptscriptstyle{\sharp}}(y) \in
\conv{\{\xhat{}{i}\}_{i=1}^N}$
for every $y$, defining the particle cloud diameter
$D(z) \triangleq
\diam{\{\xhat{}{i}\}_{i=1}^N}$,
the gradient is uniformly bounded for every $y$ and $j$:
\begin{equation}\label{eq:lip}
  \norm{\nabla_y \log r\upp{j}(y)}
  \leq \norm{\errcov^{-1}} D(z).
\end{equation}
Denote this Lipschitz bound by
$L(z) \triangleq \norm{\errcov^{-1}} D(z)$ and consider the function
\begin{equation}
  \psi(z) \triangleq 
  B(z,\bar{y}(z))
  \geq 0,\nonumber
\end{equation}
where
\begin{equation}
  \bar{y}(z) \triangleq
  \arg\min_y \max_j \norm{y - \xhat{}{j}}\nonumber
\end{equation}
is the Chebyshev center~\cite[Chapter 8.5.1]{boyd2004convex} of the particle cloud.
%
%
\begin{lemma}\label{lem:bayesian} For any 
$\delta_1,\mu\in(0,1)$, define
\begin{align*}
    \Delta_b=\Delta_{b}(\delta_1,\mu)\triangleq
    \mu A_1 + (1-\mu)B_1,
\end{align*}
where
\begin{align*}
    A_1 &\triangleq \sup_{z\in\mathcal{Z}} 
        3\psi(z) + 3L(z)
        \Bigl(\bdd \Delta t+\kappa_n(\delta_1)\Bigr)\nonumber\\
    B_1 &\triangleq \sup_{z\in\mathcal{Z}} 3L(z)\norm{x_{q,\theta\com}(t_{k+1})- \bar{y}(z)}\nonumber\\
    \kappa_n(\delta_1) &\triangleq \sqrt{\norm{\errcov}(n+2\sqrt{n\log \tfrac{1}{\delta_1}}-2\log \delta_1)}
\end{align*}
Then for every $z \in \mathcal{Z}$, there exists $u \in\RR^n$ such that
\begin{equation}
  \prob{b(\zsharp)
    \geq b(z) - \Delta_b
     \big| z,u,x_k}
  \geq 1 - \delta_1.\nonumber
\end{equation}
\end{lemma}
\begin{proof} Let $\mu\in(0,1)$.
Select $u_p$ to be the privacy controller such that $x_k + \Delta t u_p = \bar{y}(z)$.
Then the controller $u_b$ for maintain privacy while following the reference path is defined as\footnote{Note that $\mu=0$ corresponds to no effort at obfuscation by the agent.}
\begin{align}\label{eqn:information-safe controller}
    u_b \triangleq \mu u_p + (1-\mu) (x_{q,\theta\com}(t_{k+1}) - x_k)/\Delta t.
\end{align}
Under the agent dynamics~\eqref{eqn:dynamics}, we have
\begin{align*}
    \norm{x_{k+1} - \mu \bar{y} - (1-\mu)x_{q,\theta\com}(t_{k+1})}\leq \bar{d}\Delta t
\end{align*}
Then, the distance between the observation and the Chebyshev center satisfies
\begin{align*}
    \norm{y - \bar{y}(z)} \leq (1-\mu)\norm{x_{q,\theta\com}(t_{k+1}) - \bar{y}} +\bdd \Delta t + \norm{\xi}.
\end{align*}
By~\eqref{eq:rsp_bound} and the Lipschitz bound~\eqref{eq:lip}, 
\begin{align*}
|b(\zsharp) - b(z)| \leq 3\left[\psi(z) + L(z)\norm{y-\bar{y}(z)}\right].
\end{align*}
By a computation following~\cite[Lemma 1]{laurent2000adaptive}, one has 
\begin{align*}
    \norm{\xi} \leq \sqrt{\norm{\errcov}(n + 2\sqrt{n\log \tfrac{1}{\delta_1}}-2\log(\delta_1))} = \kappa_n(\delta_1)
\end{align*}
with probability at least $1-\delta_1$.
It follows that $|b(\zsharp) - b(z)| \leq \Delta_{b}$, hence $b(\zsharp) \geq b(z) - \Delta_{b}$.
\end{proof}

\begin{remark}\label{remark3}
Note that, to enforce the tracking error condition \eqref{eqn:envelope constraint}, the value of $\mu$ is required to satisfy $\norm{x_{k+1} - x_{q,\theta\com}(t_{k+1})}\leq \varrho(\ts{k+1})$, which will be guaranteed by:
\begin{align}\label{eqn:mu upper bound}
    \mu\leq
    \mu_{max}(k)\triangleq\min\left\{
        1,    \frac{\varrho(\ts{k+1}) - \bdd \Delta t}{\norm{x_{q,\theta\com}(t_{k+1}) - \bar{y}}}
    \right\}.
\end{align}
To see this, use the triangle inequality in
\begin{align*}
    &\norm{x_{k+1} -x_{q,\theta\com}(t_{k+1})}\\
    &\leq\norm{x_{k+1} - \mu \bar{y} - (1-\mu)x_{q,\theta\com}(t_{k+1})}\\
    &+\norm{\mu (\bar{y}- x_{q,\theta\com}(t_{k+1}))}\\
    &\leq \bar{d}\Delta t +\mu\norm{(\bar{y}- x_{q,\theta\com}(t_{k+1}))}
\end{align*}
The value of $\mu_{max}$ is obtained by requiring the right-hand side to be less than or equal to $\varrho(\ts{k+1})$.
\end{remark}

\subsection{Resampling Update Barrier}\label{subsec:resampling keep lower bound}
The resampling step in the RBPF update (Section~\ref{intent inference}) introduces a discontinuous change in the empirical posterior when $\zsharp\in \mathcal{Z}$ satisfies $\neff=\neff(\zsharp)\leq N_0$, recall~\eqref{eqn:effective sample size}.
Denote the set of all such $\zsharp$ by $\mathcal{Z}_0$.

Particles outside of $\rsmajor{}\subseteq \allparticles$, the set of $\neff$ highest weight particles, are discarded.
The rest are retained in the filter and replicated with weights $\tfrac{1}{N}$, each in $N_a=\lfloor \rswgt{}{a} N / \rswgteff \rfloor$ copies, where $\rswgteff=\sum_{a \in \rsmajor{}} \rswgt{}{a}$.
Let $\major{}$ denote the set of indices of replicated effective particles. 
$\major{}$ has cardinality $\sum_{a \in \rsmajor{}} N_a$.
The remaining particles are reinitialized by sampling from a distribution $\particledistro$ over $\Theta$.
After resampling, all weights are assigned uniformly, $\wgt{}{i} = \tfrac{1}{N}$ for all $i\in\allparticles$.
If no resampling is triggered ($\neff> N_0$), then set $\major{}:=\rsmajor{}$ and $\wgt{}{i} := \rswgt{}{i}$ for all $i \in \allparticles$.

Recalling~\eqref{eqn:weighted avg estimator post resampling}, let $\major{}\comp=\allparticles\minus\major{}$, with $N'\triangleq|\major{}\comp|\leq N_0$.
%
The post resampling distribution $\hqavg_{ }$ can be written as:
\begin{align*}
    \hqavg_{ }\triangleq
    \tfrac{1}{N}\textstyle\sum_{a\in\major{}}\qrep{}{a} + \tfrac{1}{N} \textstyle\sum_{b\in\major{}\comp} \qrep{}{b}.
\end{align*}
%
%
%
Similarly to~\eqref{eqn:barrier_change}, a probabilistic bound on the barrier value difference
\begin{align}\label{eqn:second barrier change}
    b(z')-b(\zsharp)=
    \textstyle\sum_{\nu\in\{x,r,t\}}\log\frac{S_\nu(\zsharp)}{S_\nu(z')}
\end{align}
is being sought.
One has:
\begin{align*}
    N\cdot S_\nu(z')
    &= \underbrace{\sum_{a \in\major{}} \gamma_\nu(\theta\upp{a})}_{\triangleq A_{\nu}}+\underbrace{\sum_{b \in \major{}\comp} \gamma_\nu(\theta\upp{b})}_{\triangleq Y_{\nu}}.
\end{align*}
Note that $A_{\nu}$ is a  deterministic function of $\zsharp$, while $Y_{\nu}$ a random variable.
Define $\bar{Y}_{\nu}\triangleq\tfrac{1}{N_0}Y_{\nu}\leq\tfrac{1}{N'
}Y_{\nu}\in(0,1)$.
%
\begin{lemma}[Resampling barrier preservation]\label{lem:resampling_conservative}
For any $\delta_2\in (0,1)$, let $\epsilon=\epsilon(\delta_2) \triangleq \sqrt{\frac{1}{2N_0}\log \frac{3}{\delta_2}}$ and define
\begin{align*}
    \Delta_r
    \triangleq \sup_{\zsharp\in\mathcal{Z}_0}\sum_{\nu \in \{x,r,t\}} 
    \log\Bigl(\frac{A_{\nu}+N_0\bigl(\mathbb{E}_{\particledistro}(\bar{Y}_{\nu})+\epsilon\bigr)}
    {N S_\nu(\zsharp)}\Bigr).
\end{align*}
Note that $\Delta_r$ depends on $\delta_2$.
Then, for any $\zsharp\in\mathcal{Z}_0$,
\begin{align*}
    \prob{b(z') \geq b(\zsharp) - 
    \Delta_r\big|\zsharp} \geq 1 - \delta_2
\end{align*}
holds for the post-resampling state $z' = R(\zsharp, \zeta)$.
\end{lemma}
\begin{proof} If $\zsharp\in\mathcal{Z}$ is such that $\neff(\zsharp)\leq N_0$, then H\"{o}ffding's inequality yields
\begin{align*}
    \prob{
        \bar{Y}_{\nu}\geq
        \mathbb{E}_{\particledistro}(\bar{Y}_{\nu}) +\epsilon
        \big|\zsharp
    }\leq \exp(-2N_0\epsilon^2)
    =\frac{\delta_{2}}{3}
\end{align*}
Then, by the independence in~\eqref{eqn:probabilistic representation of intent}, the event
\begin{align*}
    \mathscr{E} &\triangleq\bigcap_{\nu\in\{x,r,t\}} \{\bar{Y}_{\nu}\leq \mathbb{E}_{\particledistro}(\bar{Y}_{\nu}) +\epsilon)\}
\end{align*}
satisfies $\prob{\mathscr{E}}\geq(1-\tfrac{\delta_2}{3})^3\geq 1-3\tfrac{\delta_2}{3}$.
Since $\log$ is monotone increasing and $N'\leq N_0$, the equality~\eqref{eqn:second barrier change} gives rise to
\begin{align*}
    b(z')-b(z\sh) &= \sum_\nu \log \left( \frac{NS_{\nu}(\zsharp)}{A_{\nu} +N_0\bar{Y}_{\nu}}
    \right)\\
    &\geq\sum_\nu \log \left( \frac{NS_{\nu}(\zsharp)}{A_{\nu}+N_0(\mathbb{E}_{\particledistro}(\bar{Y}_{\nu})+\epsilon)}
    \right)
\end{align*}
being valid over the event $\mathscr{E}$.
This proves the Lemma.
\end{proof}

    %
    %

\subsection{Composite PCBF Bound}\label{subsec:privacy_spec}
Both Lemma~\ref{lem:bayesian} and Lemma~\ref{lem:resampling_conservative} 
establish additive lower bounds on the barrier change~\eqref{eqn:barrier_change} with high probability, rather than on the multiplicative form required by the PCBF condition~\eqref{eqn:pcbf_one_step}.
The multiplicative form is recovered by restricting to the sublevel set 
\begin{equation}
    \CIbeta \triangleq \{z \in \mathcal{Z} : b(z) \geq \beta\}.\nonumber
\end{equation}
Indeed, for any $\Delta\in(0,\beta)$ and $z\in\CIbeta$ one has
\begin{align}\label{eqn:additive to multiplicative conversion}
b(z) - \Delta &= b(z)\Bigl(1 - \tfrac{\Delta}{b(z)}\Bigr) \geq 
b(z)\Bigl(1 - \tfrac{\Delta}{\beta}\Bigr) = \alpha b(z),
\end{align}
where $\alpha\triangleq 1 - \tfrac{\Delta}{\beta} \in (0,1)$.
The following theorem applies this conversion to compose the 
two additive bounds into a single PCBF condition.

%
%
%
%
%
\begin{theorem}[Composite information PCBF]\label{thm:composite}
Let $\delta_1, \delta_2\in (0,1)$ and define $\delta_f \triangleq 1 - (1-\delta_1)(1-\delta_2)$.
Let $\Delta_{tot} \triangleq \Delta_{b} + \Delta_r$ and let $\beta > \Delta_{tot}$.
Define $\alpha \triangleq 1 - \tfrac{\Delta_{tot}}{\beta}\in (0,1)$.
Then for every $z \in \CIbeta$, there exists $u \in\RR^n$ such that
\begin{equation}
  \prob{b(F_I(z,u,\eta)) \geq \alpha b(z) \middle| z, u} \geq 1 - \delta_f.\nonumber
\end{equation}
\end{theorem}

\begin{proof}
Fix $z \in \CIbeta$, so $b(z) \geq \beta$.
Let $u$ be the control whose existence is guaranteed by Lemma~\ref{lem:bayesian}.
Define the events
\begin{align*}
  E_1 &\triangleq \bigl\{b(\Fsharp(z,u,\xi)) \geq b(z) - \Delta_{b}\bigr\}, \\
  E_2 &\triangleq \bigl\{b(R(\zsharp,\zeta)) \geq b(\zsharp) - \Delta_{r}\bigr\}.
\end{align*}
By Lemma~\ref{lem:bayesian}, $\prob{E_1 \mid z, u} \geq 1 - \delta_1$.
By Lemma~\ref{lem:resampling_conservative}, $\prob{E_2 \mid \zsharp} \geq 1 - \delta_2$.
Since $\xi$ and $\zeta$ are independent,
\begin{equation*}
  \prob{E_1 \cap E_2 \mid z, u} \geq (1-\delta_1)(1-\delta_2) = 1 - \delta_f.
\end{equation*}
On the event $E_1 \cap E_2$, the chain of inequalities
\begin{align*}
  b(F_I(z,u,\eta)) &= b(R(\zsharp,\zeta))\geq b(\zsharp) - \Delta_{r} \geq b(z) - \Delta_{tot}
\end{align*}
holds.
Since $b(z) \geq \beta > \Delta_{tot}$, applying~\eqref{eqn:additive to multiplicative conversion} completes the proof.
\end{proof}

\begin{remark}[Separation of privacy and tracking constraints]\label{rem:separation}
  Theorem~\ref{thm:composite} establishes the existence of a feasible input $u \in U$ satisfying the PCBF condition without reference to the tracking envelope $\varrho$.
  %
  %
  When both constraints can be simultaneously satisfied, the feasible controller set is nonempty and the agent can maintain privacy while tracking.
  When the tracking envelope becomes too restrictive, the feasible set of the joint program may be empty.
\end{remark}

\section{Discussion: Joint Feasibility}\label{discussion}
Theorem~\ref{thm:composite} guarantees that for every $z \in \CIbeta$, $\beta > \Delta_{tot}$, at each time $\ts{k}$ there exists at least one control input satisfying the PCBF privacy condition.
However, according to Remark~\ref{remark3}, the success of this control input in maintaining the tracking envelope constraint is dependent on our ability to select $\mu$ at time $k$ according to~\eqref{eqn:mu upper bound}.

\begin{figure*}
    \centering
    \includegraphics[width=\linewidth]{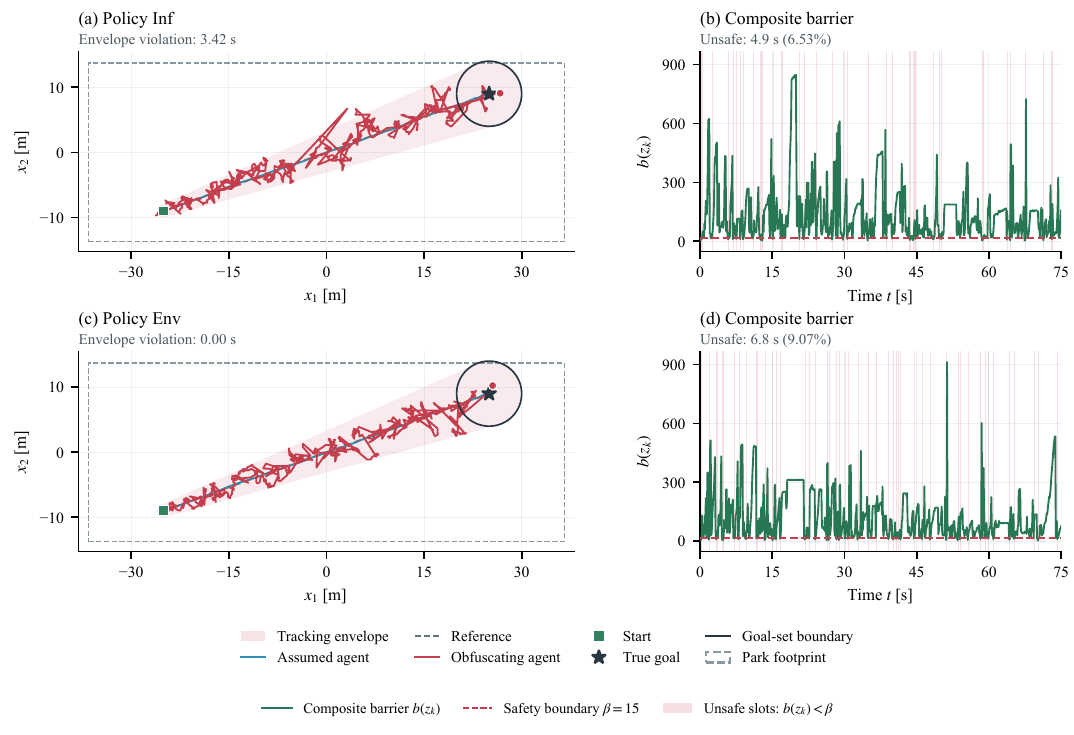}
    \caption{Side-by-side comparison of the proposed control policies $\mathbf{Inf}$ (top, ``information safety first'') and $\mathbf{Env}$ (bottom, ``envelope first''), with $\gamma=5$, random seed $2048$.
    Left: trajectories for obfuscating agent (red) and assumed agent (blue); reference path (dashed); and the tracking envelope (pink).
    The circle denotes the goal set.
    Right: the information lower bound $b(z_k)$ as a function of time.
    The dashed line marks $\beta$, and translucent red bands mark complete time slots for which $b(z_k)<\beta$.
    The unsafe-time fractions are $6.53\%$ for $\mathbf{Inf}$ and $9.07\%$ for $\mathbf{Env}.$}
    \label{fig:placeholder}
\end{figure*}

The requirement that $\beta>\Delta_{tot}$ represents an additional tension between the problem constraints.
On one hand, $\beta$ is provided, essentially, by the initial condition, as it must be a lower bound on $\kldiv{q_{\theta\com}}{\hqavg_0}$ (the agent is initialized in $C_I^\beta$).
On the other hand, at each of the times $\ts{k}$ one has
\begin{align*}
    \Delta_{tot}=\Delta_r+\Delta_b=\Delta_r+\mu A_1+(1-\mu)B_1,
\end{align*}
where $\mu$ is allowed to depend on $k$, and influences the control input being selected.
Thus it is necessary to always be selecting $\mu$ that guarantees $\Delta_{tot}<\beta$.
Equivalently,
\begin{align}\label{eqn:information safe controller}
    \beta-\Delta_r-B_1 > \mu(A_1-B_1).
\end{align}
If $A_1-B_1\geq 0$ at time $k$, the resulting upper bound combines with~\eqref{eqn:mu upper bound} to produce a value for $\mu$ (and through it, a control input maintaining both mission constraints); if $A_1-B_1<0$ results in a positive lower bound on $\mu$, which may or may not be consistent with~\eqref{eqn:mu upper bound}. 
%
%

\noindent{\bf Proposed Control Policies.}
It follows from the preceding discussion that selecting a control input at time $\ts{k}$ according to~\eqref{eqn:information safe controller} is always possible when $A_1-B_1\geq 0$ using\footnote{Note that the case $A_1-B_1=0$ does not pose a problem, since it simply corresponds to selecting $\mu=\mu_{max}$.}
\begin{align*}
    \mu=\min\left(
        \mu_{max},
        \tfrac{\beta-\Delta_r-B_1}{A_1-B_1}
    \right)-\varsigma,
\end{align*}
where $\varsigma>0$ is a small parameter.
When $A_1-B_1<0$, the controller~\eqref{eqn:information safe controller} satisfies the joint constraints only if
\begin{align*}
    \tfrac{\beta-\Delta_r-B_1}{A_1-B_1}+\varsigma\leq\mu\leq\mu_{max},
\end{align*}
prompting the choice of $\mu=\mu_{max}$ whenever this feasibility interval is non-empty.
However, when it is empty, the envelope and information constraints are in conflict, prompting a control input selection according to one of two policies:
policy $\mathbf{Inf}$ prioritizes information safety~\eqref{eqn:information safe controller} by selecting $\mu=\tfrac{\beta-\Delta_r-B_1}{A_1-B_1}+\varsigma$ to enable at least some drift towards the goal;
policy $\mathbf{Env}$ selects $\mu=\mu_{max}$, prioritizing the envelope constraint, as expressed in~\eqref{eqn:mu upper bound}, while allowing some obfuscation, however insufficient.

\section{Numerical Experiment}
The trade-offs between information safety and tracking were investigated numerically by simulating the two proposed control policies.

\smallskip
\noindent{\bf Setup and Parameter Selection.} The agent starts at $x_0=(-25,-9)^\top\,\mathrm{m}$ and has intent $\theta^*=(x^*,r^*,t^*)$, where $x^*=(25,9)^\top\,\mathrm{m}$, $r^*=5\,\mathrm{m}$, and $t^*=T=75 \mathrm{s}$.\footnote{The spatial scale of the simulation is selected to match the dimensions of the University of Florida Autonomy Park, to simulate a realistic setting.
See \url{https://mae.ufl.edu/research/facilities/autonomy-park/}.}
Controls are held constant over $\Delta t=0.1\,\mathrm{s}$, with $x_{k+1}=x_k+\Delta t(u_k+d_k)$, where $d_k$ is sampled independently with $\bar{d} = 0.5$.
The observer receives position measurements with independent Gaussian noise of standard deviation $0.1\,\mathrm{m}$ per coordinate.
Its RBPF uses $N=2000$ particles and resamples when the effective sample size falls below $500$.
The prior is uniform over $55\mathbb{B}\times[1,7]\times[10,85]$, with spatial quantities in meters and arrival times in seconds.
The threshold $\gamma=5$ is fixed before control synthesis, defining $\mathcal{C}_I=\{z:\underline{H}(z)\geq\gamma\}$.
Note, the small constant $\varsigma$ is omitted from the numerical experiment.

\smallskip
\noindent{\bf Results.}
Figure~\ref{fig:placeholder} shows a paired realization with a shared initial particle cloud and matched disturbance and measurement streams.
Policy $\mathbf{Inf}$ spends $4.9\mathrm{s}$ ($6.53\%$) outside $\mathcal{C}_I$ and $3.42\mathrm{s}$ outside the tracking envelope.
Policy $\mathbf{Env}$ never breach the envelop, while its unsafe duration increases to $6.8\mathrm{s}$ ($9.07\%$).
Thus, policy $\mathbf{Inf}$ meets the $5\%$ unsafe-time target in this realization, whereas policy $\mathbf{Env}$ meets the relaxed $10\%$ target.
Both agents lie in the goal set at $T$.
Table \ref{tab:gamma10} gives the detail statistical result of all simulations.
These results illustrate the trade-off between information safety and tracking.
\begin{table}[t]
  \centering
  \caption{Observed constraint-violation probabilities as fractions of experiment duration, $\beta=15$.}
  \label{tab:gamma10}
  \small
  \setlength{\tabcolsep}{5pt}
  \begin{tabular}{p{1in}p{.8125in}p{.9in}}
    \toprule
    Policy
    & Barrier violation (\%)
    & Envelope violation (\%) \\
    \midrule
    $\mathbf{Inf}$:\hfill safety first
    & $7.03 \pm 0.97$
    & $5.46 \pm 1.98$ \\
    $\mathbf{Env}$:\hfill envelope first
    & $8.24 \pm 0.87$
    & $0.00 \pm 0.00$ \\
    \bottomrule
  \end{tabular}
\end{table}

\section{Conclusion}\label{conclusion}
This paper addresses the problem of intent privacy for a goal-directed agent operating in the presence of a passive adversary running an online Bayesian inference algorithm to infer the agent's goal.
%
%
Rather than modify the task or treat privacy as an auxiliary penalty, intent privacy was formulated as a control problem with piecewise-constant actuation over a state space extended by the latent belief state, while treating the RBPF information state $z = (\theta, \hat{x}, \omega)$ as a controlled Markov process and the observer's KL-based information leakage as the quantity to be regulated.

The main technical contribution is a composite probabilistic control barrier function result that provides finite-horizon privacy-in-probability guarantees for the full RBPF update cycle together with an explicit piecewise-constant controller realizing this guarantee.
The Bayesian update step produces an additive lower bound on the barrier increment by controlling the observation $y$ toward the Chebyshev center of the particle cloud.
The resampling step produces a second additive lower bound by applying H\"{o}ffding's inequality to the random contribution of reinitialized particles.
The two bounds compose multiplicatively through a union bound argument, yielding a single PCBF condition covering the entire update cycle with failure probability $\delta_f$.

A complementary contribution exposes tensions between the obfuscation component of the mission and its tracking component.
It is demonstrated that large uncertainties of the RBPF representation of the the agent's state can render the joint mission infeasible in the presence of tight tracking envelope constraints.
At the same time, feasibility is recovered if the envelope constraints are relaxed, and the relation is quantified.
Two control strategies are introduced to overcome such situations, and a numerical study of their effectiveness is provided.

Several directions remain open for future investigation, not least among them is obtaining convergence results which characterize the rates of convergence of the RBPF intent-filtering framework.
Such results are inextricably linked with the joint feasibility of the tracking and obfuscation problems (see Lemma~\ref{lem:bayesian}), being pivotal for establishing a precise quantification of the trade-offs governing the natural growth rates of tracking envelopes permissible for effective obfuscation and further development of theoretical guarantees for the strategies introduced for the infeasible case.
Further work establishing these trade-offs in continuous time is also necessary for safe real-life deployment.

\section{Acknowledgments}
The authors thank the reviewers for their careful reading of the manuscript and for their constructive suggestions, which substantially improved the clarity and rigor of our presentation.

\bibliographystyle{IEEEtran}
\bibliography{references.bib}

\end{document}